\documentclass[11pt,notitlepage]{article}
\usepackage[utf8]{inputenc}
\usepackage[margin=1in]{geometry}
\usepackage[onehalfspacing]{setspace}

\usepackage{natbib}
\usepackage{url}
 
\usepackage{graphicx}
\usepackage{amssymb}
\usepackage{amsmath}
\usepackage{amsthm}

\newtheorem{lemma}{Lemma}
\newtheorem{theorem}{Theorem}

\newcommand{\R}{\ensuremath{\mathbb{R}}}

\newcommand{\E}{\ensuremath{\mathbb{E}}}

\newcommand{\dis}{\text{dis}}
\newcommand{\agg}{\text{agg}}
\newcommand{\com}{\text{com}}

\def\b1{\boldsymbol{1}}

\newenvironment{figurenotes}[1][Note]{\begin{minipage}[t]{\linewidth}\footnotesize{\itshape#1: }}{\end{minipage}}
\begin{document}
 
\title{Temporal Aggregation for the Synthetic Control Method\thanks{Prepared for 2024 AEA P\&P session ``Treatment Effects: Theory and Implementation".}}
 
\author{Liyang Sun, Eli Ben-Michael, and Avi Feller\thanks{Sun: Department of Economics, University College London and CEMFI, Email: liyang.sun@ucl.ac.uk. 
Ben-Michael: Department of Statistics \& Data Science and Heinz College of Information Systems \& Public Policy, Carnegie Mellon University. Feller: Goldman School of Public Policy \& Department of Statistics, University of California, Berkeley. We thank Elizabeth Stuart for sharing data for the application. AF and LS gratefully acknowledge support from the Institute of Education Sciences, U.S. Department of Education, through
Grant R305D200010. LS also acknowledges support from Grant PID2022-143184NA-I00 funded by MCIN/AEI. See the associated \texttt{augsynth} library for \texttt{R} and \protect\url{https://doi.org/10.5281/zenodo.10848594} for the replication package.}}

\maketitle
\begin{abstract}
The synthetic control method (SCM) is a popular approach for estimating the impact of a treatment on a single unit with panel data. Two challenges arise with higher frequency data (e.g., monthly versus yearly): (1) achieving excellent pre-treatment fit is typically more challenging; and (2) overfitting to noise is more likely. Aggregating data over time can mitigate these problems but can also destroy important signal. In this paper, we bound the bias for SCM with disaggregated and aggregated outcomes and give conditions under which aggregating tightens the bounds. We then propose finding weights that balance both disaggregated and aggregated series.
\end{abstract}

\section{Introduction}
Empirical researchers often use the synthetic control method (SCM) to estimate the impact of a treatment on a single unit in panel data settings \citep{AbadieAlbertoDiamond2010}. The ``synthetic control'' is a weighted average of control units that balances the treated unit’s pre-treatment outcomes as closely as possible.

Two challenges arise when using SCM with higher frequency data, such as when the outcome is measured every month versus every year. First, because there are more pre-treatment outcomes to balance, achieving excellent pre-treatment fit is typically more challenging. Second, even when excellent pre-treatment fit is possible, higher-frequency observations raise the possibility of bias due to overfitting to noise. A recent review by \citet{abadie2022synthetic} explicitly cautions about such bias from using disaggregated outcomes in SCM. Instead, researchers can first aggregate the outcome  series into lower-frequency (e.g., annual) observations, and then estimate SCM weights that minimize the imbalance in these aggregated pre-treatment outcomes. Doing so mechanically improves pre-treatment fit as well.

In this paper, we propose a framework for temporal aggregation for SCM. Adapting recent results from \citet{sun2023using}, we first derive finite-sample bounds on the bias for SCM under a linear factor model when using temporally disaggregated versus aggregated outcome series.
With these bounds, we then show that temporal aggregation 
only reduces bias to the extent that doing so reduces noise without also overly reducing signal for the underlying factors.
The optimal trade-off between the two approaches depends on unknown parameters. We argue, however, that finding synthetic control weights that jointly balance both the disaggregated and aggregate series is a promising compromise approach for practice.

Our setup builds on an expansive literature on SCM, especially recent papers that modify SCM to mitigate bias both due to imperfect pre-treatment balance \citep[e.g.,][]{ferman2018revisiting, benmichael2021_ascm}
and bias due to overfitting to noise \citep[e.g.,][]{kellogg2021combining}. Most directly relevant, \citet{sun2023using} discuss using use multiple outcomes to mitigate both sources of bias.

\section{Setup}

For each unit $i=1,\dots,N$ and at each lower-frequency time interval $t=1,\dots,T$,  we observe $K$ higher-frequency observations of the outcome denoted as $Y_{itk}$ where $k=1,\dots,K$. For instance, in a monthly series $t=1,\dots,T$ could represent years and $k=1,2,\dots,12$ the months within each year. The choice of time interval $t$ reflects how we compute temporal aggregates, such as aggregating monthly data into yearly averages. We maintain fixed values for $N$ and $K$ throughout the analysis.

We denote the exposure to a binary treatment by $W_{i}\in\{0,1\}$. We restrict our attention to the case where a single unit receives treatment, and follow the convention that this is the first one, $W_1 = 1$. The remaining $N_0 \equiv N-1$ units are possible controls, often referred to as ``donor units.'' To simplify notation, we limit to one post-treatment period, $T = T_0 + 1$, though our results easily extend to larger $T$.

We  denote the potential outcome under treatment $w$ with $Y_{itk}(w)$. We are interested in the treatment effects for the treated unit during the $K$ higher frequency observations during post-treatment period $T$: $\tau_{Tk} = Y_{1Tk}(1) - Y_{1Tk}(0)$ for $k=1,\dots,K$.  Since we directly observe $Y_{1Tk}(1) = Y_{1Tk}$ for the treated unit, we focus on imputing the missing counterfactual outcome under control, $Y_{1Tk}(0)$.

Throughout, we will focus on \emph{de-meaned} or \emph{intercept-shifted} weighting estimators \citep{Doudchenko2017, ferman2018revisiting}.
We denote $\bar Y_{i\cdot \cdot} \equiv \frac{1}{T_0K}\sum_{t=1}^{T_0}\sum_{j=1}^{K} Y_{itj}$ as the pre-treatment average for the outcome for unit $i$, and $\dot Y_{itk} = Y_{itk} - \bar{Y}_{i\cdot \cdot}$ as the corresponding de-meaned outcome.
We consider estimators of the form: $\widehat{Y}_{1Tk}(0) \equiv \bar{Y}_{1\cdot \cdot} + \sum_{i=2}^N \gamma_i \dot{Y}_{iTk},$ 
where $\gamma \in \mathcal C \subset \R^{N-1}$ is a set of weights. Our paper centers on how to choose the weights $\gamma$ from a set $\mathcal{C}=\{\gamma\in\mathbb{R}^{N-1}\mid\|\gamma\|_1\leq C,\sum_{i}\gamma_{i}=1\}$ for a known $C$.

The first approach we consider is finding a synthetic control that has the best pre-treatment fit on the (de-meaned) disaggregated high-frequency outcomes $q^{\dis}(\cdot)$.
We refer to this set of weights 
as the \emph{disaggregated weights} $\hat{\gamma}^{\dis}$:
\[
 \underset{\gamma\in  \mathcal{C}}{\min}  \frac{1}{T_{0}}\frac{1}{K}\sum_{k=1}^{K}\sum_{t=1}^{T_{0}}\left(\dot Y_{1tk} -\sum_{W_{i}=0}\gamma_{i} \dot Y_{itk}\right)^{2}.
\]

An alternative choice is to optimize the \emph{aggregated objective} $q^{\agg}(\cdot)$, the pre-treatment fit for the temporally aggregated outcomes via averaging. We refer to the set of weights that minimize this objective as the \emph{aggregated} weights $\hat{\gamma}^{\agg}$:
\[
\underset{\gamma\in  \mathcal{C}}{\min} \frac{1}{T_{0}}\sum_{t=1}^{T_{0}}\left(\frac{1}{K}\sum_{k=1}^{K}\dot Y_{1tk} -\sum_{W_{i}=0}\gamma_{i}\dot Y_{itk}\right)^{2}.
\]

\section{Bias bounds}
To derive finite sample bias bounds for $Y_{1Tk}(0)  - \widehat{Y}_{1Tk}(0)$ for each $k=1,\dots,K$, we assume that the outcomes under control are generated as
$Y_{itk}(0) = \alpha_{i} + 
\beta_{tk} + L_{itk} + \varepsilon_{itk}$, with $\sum_{t=1}^T \beta_{tk} = 0$ for all $k$; see, for example, \citet{athey_matrix_2021}.
After incorporating the additive two-way fixed effects, the model component retains a term $L_{itk}$ with $\sum_{i=1}^N L_{itk} = 0$ for all $t, k$ and $\sum_{t=1}^T L_{itk} = 0$ for all $i, k$.
We assume the idiosyncratic errors $\varepsilon_{itk}$  are mean zero sub-Gaussian random variables with scale parameter $\sigma$, independent of the treatment status
$W_{i}$. We also assume independence across units and time, which is plausible if the model components capture co-movement in the outcome.

Let the matrix $L\in\mathbb{R}^{N\times(TK)}$ contain  $L_{itk}$ 
 for the treated unit and the remaining rows correspond to control units. \cite{sun2023using} show that  a low rank condition $rank(L)<N-1$ is necessary for there to exist  oracle weights $\gamma^\ast$ over donor units that yield an unbiased estimate for the control potential outcome: $ \E_{\varepsilon_{Tk}}\left[Y_{1Tk}(0)  - \widehat{Y}_{1Tk}(0) \right] = L_{1Tk} - \sum_{i = 2}^N \gamma^\ast_i L_{iTk} = 0. $
Here the expectation is taken over the idiosyncratic errors in the respective post-treatment periods. 

The deterministic model component can be written as a linear factor model, 
$L_{itk}=\boldsymbol{\phi}_{i}\cdot\boldsymbol{\mu}_{tk},$ 
with $r=rank(L)$, where $\boldsymbol{\mu}_{tk}\in\mathbb{R}^{r}$ are latent time factors and each unit has a vector of time invariant factor loadings
$\boldsymbol{\phi}_{i}\in\mathbb{R}^{r}$.  
Two important quantities for our discussion are $\underbar{\ensuremath{\xi}}^{\dis}$ and $\underbar{\ensuremath{\xi}}^{\agg}$, the smallest singular values of the variance-covariance matrix of, respectively, the time factors $\mu_{tk}$ and the averaged time factors $\bar{\mu}_{t}=\frac{1}{K}\sum_{k=1}^{K}\mu_{tk}$.
Following previous literature we assume that $\underbar{\ensuremath{\xi}}^{\dis} > 0$, which avoids issues of weak identification \citep{AbadieAlbertoDiamond2010}.

 For estimated weights $\hat{\gamma}$ based on pre-treatment fit, the bias for the effect in period $Tk$ is due to inadequate balance in the model components:
$
    Bias(\hat\gamma)=L_{1Tk}- \sum_{i=2}^N \hat{\gamma}_i L_{iTk} .
$
We can decompose the bias from estimated weights $Bias(\hat\gamma)$
into two terms using the linear factor model:
\begin{align*}
  & \sum_{t=1}^{T_{0}}\sum_{j=1}^{K}{\omega}_{tj}\left(\dot{Y}_{1tj}-\sum_{W_{i}=0}\hat{\gamma}_{i}\dot{Y}_{itj}\right)\ \ (\text{imbalance})\\
 & -\sum_{t=1}^{T_{0}}\sum_{j=1}^{K}{\omega}_{tj}\left(\dot{\varepsilon}_{1tj}-\sum_{W_{i}=0}\hat{\gamma}_{i}\dot{\varepsilon}_{itj}\right)\ \ (\text{overfitting})
\end{align*}
where  $\omega_{tj}$ are transformations of the factor values that depend on the estimator.

The first term is bias due to imperfect pre-treatment fit (or imbalance) in the pre-treatment outcomes, $\dot{Y}_{itj}$.
The second term is bias due to overfitting to noise, also known as the approximation error.
This term arises because the optimization problems minimize imbalance in \emph{observed} pre-treatment outcomes --- noisy realizations of latent factors --- rather than minimizing imbalance in the latent factors themselves. 

Theorem 1 in the appendix formally states high-probability bounds on the bias terms, which we obtain using results from \cite{sun2023using}. 
These bounds hold in finite samples and account for imperfect pre-treatment fit. The leading terms in the bias due to imbalance and overfitting are:
\begin{align*}
  |Bias(\hat\gamma^{\dis})| = & O\left(\frac{1}{\underbar{\ensuremath{\xi}}^{\dis}}\right) +
  O\left(\frac{1} {\underbar{\ensuremath{\xi}}^{\dis}\sqrt{T_0K}}\right)\\
 |Bias(\hat\gamma^{\agg})| = & O\left(\frac{1} {\underbar{\ensuremath{\xi}}^{\agg}\sqrt K}\right) +
  O\left(\frac{1} {\underbar{\ensuremath{\xi}}^{\agg}\sqrt{T_0K}}\right)
\end{align*}
where $\underbar{\ensuremath{\xi}}^{\dis}$ and $\underbar{\ensuremath{\xi}}^{\agg}$ are the relevant smallest singular values for the disaggregated and aggregated series.

The first term, bias due to imbalance, does not vanish for either approach.
Consistent with \citet{ferman2018revisiting}, bias remains because the SCM objective does not converge to the objective minimized by the oracle weights due to noise in the outcomes.
Aggregation, however, reduces the level of noise in the objective and therefore reduces the bias by a factor of $1/\sqrt{K}$.

The second term in the bias is the contribution of overfitting to noise.
This decreases in the total number of pre-treatment observations $T_0 K$ for \emph{both} the disaggregated and aggregated estimators.
For the disaggregated weights this follows prior results \citep[e.g.,][]{AbadieAlbertoDiamond2010}.
Aggregation reduces the noise in the objective, but this is counteracted by a commensurate reduction in the number of pre-treatment observations, netting out to the same risk of bias due to overfitting as with the disaggregated weights.

In practice, the most important consideration in temporal aggregation is whether doing so eliminates the signal for the underlying factors.
In particular, if (ignoring constant terms) $\sqrt{K} \underline{\xi}^{\text{agg}} > \underline{\xi}^{\text{dis}}$, aggregation will lead to tighter bias bounds.
There are many scenarios where we expect long-run variation to persist, especially after seasonally adjusting the outcomes, and therefore believe this condition might hold.
However, if there is not substantial long-run variation and  $\underbar{\ensuremath{\xi}}^{\agg}$ is very small, aggregating can 
leave little behind to learn about the latent factor loadings, and possibly inflate the bias bound. 
Similar challenges arise in time series model estimation, where aggregation can lead to biased estimates for the true time series models \citep[see][among others]{marcet2019temporal}.

Rather than choose between these two extremes, we propose finding SCM weights that control a linear combination of the two objectives, with weight $\nu$ on the aggregated objective and weight $(1-\nu)$ on the disaggregated objective.
This creates 
an imbalance ``frontier", similar to the approach in \cite{benmichael2022_stag},
where $\nu = 0$ corresponds to the disaggregated objective and $\nu = 1$ to the aggregated objective.
As formalized in Lemma 1 in the appendix, if the SCM weights yield excellent pre-treatment fit on \emph{both} the disaggregated and aggregated outcomes, these weights will also achieve the minimum of the two bounds.
In general, finding the optimal $\nu^\ast$ involves model-derived parameters and is infeasible.
As a heuristic, we propose equal weight to the disaggregated and aggregated fits, $\nu = 0.5$, and then assessing sensitivity to this choice.

There are many directions for future work that incorporate recent innovations in panel data methods, including first de-noising \citep[e.g.,][]{amjad_robust_2018} or seasonally adjusting the disaggregated outcome series. We could also explore choosing an optimal level of temporal aggregation for a single SCM objective. Finally, questions about temporal aggregation also arise in event study and other panel data models, suggesting further avenues for fruitful research.

\section{Texas 2021 Abortion Restrictions}\label{sec:application}

We revisit the \citet{bell_texas_2023} study on the SB8, a 2021 Texas law  restricting abortion. Using monthly state-level live birth counts from 2016 to 2022, the authors  construct a synthetic Texas to estimate SB8's impact on monthly births.

\begin{figure}[!ht]
\begin{center}
\includegraphics[width=0.5\textwidth]{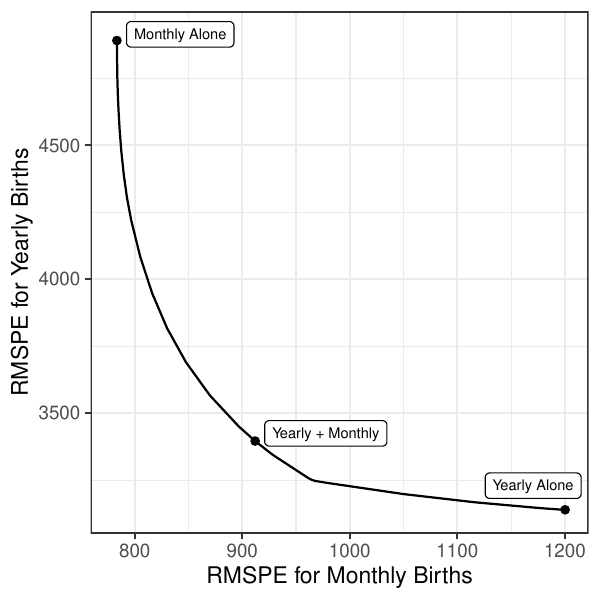}
\end{center}
\caption{Imbalance ``frontier"}\label{fig:imbalance}
\begin{figurenotes}
     The pre-treatment fit for births aggregated to the yearly level (y-axis) against the pre-treatment fit for births at the monthly level (x-axis) as we vary the relative weight placed on balancing average yearly births from 0 (top left corner) to 1 (bottom right corner). ``Yearly + Monthly'' denotes the fit when equal weight is placed on both the aggregated and disaggregated series.
\end{figurenotes}
\end{figure}

The original analysis uses SCM on monthly data ($\nu = 0$). We explore aggregating to yearly averages ($\nu=1$) and intermediate values of $\nu$ as illustrated in Figure~\ref{fig:imbalance}. Combining yearly and monthly births equally ($\nu=0.5$, labeled ``Yearly + Monthly" in Figure~\ref{fig:imbalance}) achieves substantial balance on both fronts, reducing potential bias. Appendix Figure 1 indicates a slightly larger estimated effect for post-treatment monthly births using the combined approach ($\nu = 0.5$) than the original analysis. Appendix Figure 2 demonstrates stable estimates across a wide range of $\nu$.

\newpage
\bibliographystyle{aea}
\bibliography{BibFile}
\newpage
\appendix
\setcounter{figure}{0}
\renewcommand{\thefigure}{A\arabic{figure}}

\begin{center}
    Online Appendix to the paper\\
Temporal Aggregation for the Synthetic Control Method\\
Liyang Sun, Eli Ben-Michael, and Avi Feller
\end{center}

\begin{theorem}
\label{thm:error_bounds} In addition to assumptions stated above, suppose some oracle weights exist in the set $\mathcal{C}$. Denote $\boldsymbol{\mu}_{tk}\in\mathbb{R}^{r}$ as the time  factors  and assume that they are bounded above by $M$.
    Furthermore, denoting $\sigma_\text{min}(A)$ as the smallest singular value of a matrix $A$, assume that (i)  $\sigma_{min}\left(\frac{1}{T_{0}K}\sum_{tk}\mu_{tk}\mu_{tk}'\right) \geq \underbar{\ensuremath{\xi}}^{\dis} > 0$; and (ii) $\sigma_{min}\left(\frac{1}{T_{0}}\sum_{t}\left(\bar{\mu}_{t}\right)\left(\bar{\mu}_{t}\right)'\right) \geq \underbar{\ensuremath{\xi}}^{\agg} > 0$ where $\bar{\mu}_{t}=\frac{1}{K}\sum_{k=1}^{K}\mu_{tk}$.  For any $\delta>0$, let $\tilde \sigma=\left( 2C\sqrt{\log2N_{0}} + (1+C)\delta\right)(1+1/\sqrt{T_0K})\sigma$,  with probability at least $1-8\exp\left(-\frac{\delta^{2}}{2}\right)-4\exp\left(-\frac{T_0K\delta^2}{2\sigma^2(1+C^2)}\right)$, the absolute bias satisfies the bound, 
 
\begin{align}
\left| \text{Bias}(\hat{\gamma}^{\dis}) \right| & \leq \frac{rM^{2}}{\underline{\xi}^{\dis}} \left( 4(1+C)\sigma + 2\delta + \frac{\tilde{\sigma}}{\sqrt{T_{0}K}} \right), \label{eq:dis_bias} \\
\left| \text{Bias}(\hat{\gamma}^{\agg}) \right| & \leq \frac{rM^{2}}{\underline{\xi}^{\agg}} \left( \frac{4(1+C)\sigma}{\sqrt{K}} + 2\delta + \frac{\tilde{\sigma}}{\sqrt{T_{0}K}} \right). \label{eq:agg_bias}
\end{align}

\end{theorem}
\begin{proof}
    Results follow directly from Theorem 1 of \cite{sun2023using} where $\hat{\gamma}^{\dis}$ correspond to the concatenated weights, and $\hat{\gamma}^{\agg}$ correspond to the average weights in their notation.
\end{proof}
\begin{lemma}\label{lem:imbalance}
   Suppose there exists $\nu^\ast \in [0,1]$ such that for $\hat\gamma^{\com} \in\arg\min_{\gamma\in \mathcal C} \nu^\ast q^{\dis}(\gamma) + (1-\nu^\ast) q^{\agg}(\gamma)$, we have $q^{\dis}(\gamma^{\com}) \leq q^{\dis}(\gamma^\ast) $ and $q^{\agg}(\gamma^{\com}) \leq q^{\agg}(\gamma^\ast) $ almost surely. For any $\delta>0$, with probability at least $1-8\exp\left(-\frac{\delta^{2}}{2}\right)-4\exp\left(-\frac{T_0K\delta^2}{2\sigma^2(1+C^2)}\right)$, the absolute bias satisfies the bound, 
 
\begin{align*}
\left| \text{Bias}(\hat{\gamma}^{\com}) \right| \leq \min & \left\{ \frac{rM^{2}}{\underline{\xi}^{\dis}} \left( 4(1+C)\sigma + 2\delta + \frac{\tilde{\sigma}}{\sqrt{T_{0}K}} \right),     \frac{rM^{2}}{\underline{\xi}^{\agg}} \left( \frac{4(1+C)\sigma}{\sqrt{K}} + 2\delta + \frac{\tilde{\sigma}}{\sqrt{T_{0}K}} \right) \right\}.  
\end{align*}  
\end{lemma}
\begin{proof}
    Since  $q^{\dis}(\gamma^{\com}) \leq q^{\dis}(\gamma^\ast) $ and $q^{\agg}(\gamma^{\com}) \leq q^{\agg}(\gamma^\ast) $ almost surely, either of the two bias bounds  stated in Theorem~\ref{thm:error_bounds} is a valid upper bound for the estimate based on the combined weights 
 $\hat\gamma^{\com}$. We may therefore take the minimum of the two bounds to bound $|Bias(\hat\gamma^{\com})|$.
\end{proof}

\begin{figure}[!ht]
\begin{center}
\includegraphics[width=0.5\textwidth]{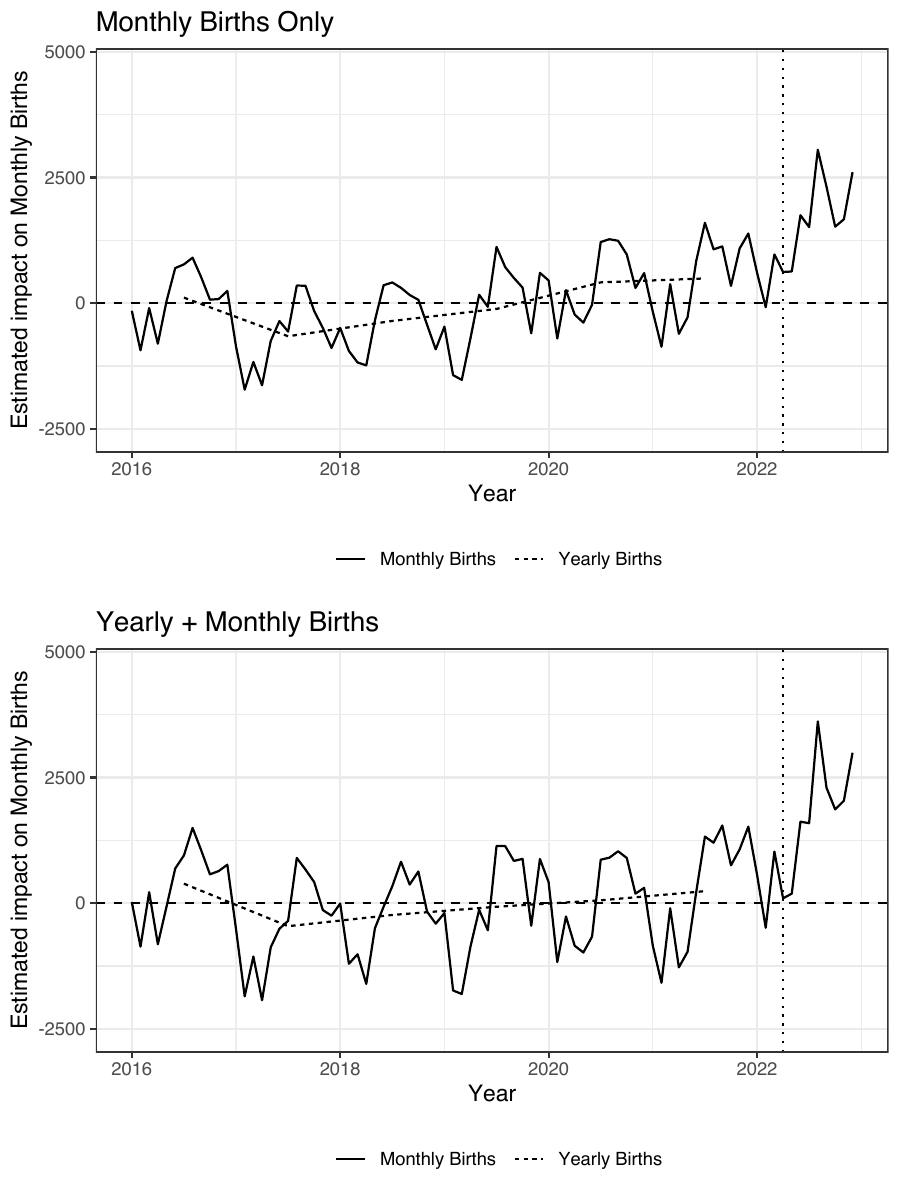}
\caption{Balance only in monthly vs in both monthly and yearly}
\label{fig:estimates}
\end{center}
\end{figure}

\begin{figure}[!ht]
\begin{center}
\includegraphics[width=0.5\textwidth]{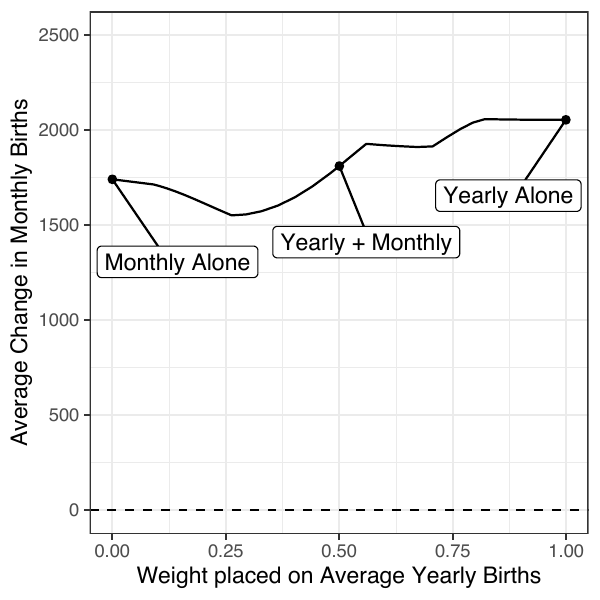}
\caption{Change in average estimated impact on monthly births as the relative weight on monthly vs yearly births varies from 0 to 1.}\label{fig:estimates_vary}
\end{center}
\end{figure} 
 
\end{document}